\documentclass[11pt]{article}
\usepackage{graphicx}
\usepackage[OT1]{fontenc}
\usepackage[latin1]{inputenc}
\usepackage[english]{babel}

\usepackage{amsmath}
\usepackage{amsthm}
\newtheorem{theorem}{Theorem}
\newtheorem{lemma}{Lemma}
\newtheorem{prop}{Proposition}

\newtheorem{definition}{Definition}
\usepackage{float}
\usepackage{geometry}
\usepackage{todonotes}

\usepackage[usenames,dvipsnames]{pstricks}
\usepackage{epsfig}
\usepackage{pst-grad} 
\usepackage{pst-plot} 
\psset{unit=0.7}
\usepackage{pstricks-add}
\usepackage{algorithm}
\usepackage{algorithmic}
\usepackage{multirow}
\usepackage{amsfonts}
\usepackage{amssymb}
\usepackage{algorithm,algorithmic}

\author{
   Eric Angel \\
   IBISC, Universit\'e d'Evry  \\
   Eric.Angel@ibisc.univ-evry.fr\\
   \and
   Evripidis Bampis\\
   LIP6, Universit\'e Pierre et Marie Curie\\
   Evripidis.Bampis@lip6.fr\\
   \and
   Vincent Chau\\
   IBISC, Universit\'e d'Evry\\
   Vincent.Chau@ibisc.univ-evry.fr
}

\date{}

\title{Throughput Maximization in the Speed-Scaling Setting\thanks{Submitted to SODA 2014}}

\begin{document}
\maketitle

\maketitle
\begin{abstract}
We are given a set of $n$ jobs and a single processor that can vary its speed dynamically.
Each job $J_j$ is characterized by its processing requirement (work) $p_j$, its release date $r_j$ and its deadline $d_j$. We are also given a budget of energy $E$ and we study the scheduling problem of maximizing the throughput (i.e. the number of jobs which are completed on time). 
We propose a dynamic programming algorithm that solves the preemptive case of the problem, i.e. when the execution of the jobs may be interrupted and resumed later,   in pseudo-polynomial time. Our algorithm can  be adapted for solving the weighted version of the problem where every job is associated with a weight $w_j$ and the objective is the maximization of the sum of the weights of the jobs that are completed on time.
Moreover, we provide a strongly polynomial time algorithm to solve the non-preemptive unweighed case when the jobs have the same processing  requirements. For the weighted case, our algorithm can be adapted for solving the non-preemptive version of the problem in pseudo-polynomial time. 
\end{abstract}

\section{Introduction}

The problem of scheduling $n$ jobs with release dates and deadlines on a single processor that can vary its speed dynamically with the objective of minimizing the energy consumption has been first studied in the seminal paper by Yao, Demers and Shenker \cite{YDS95}. In this paper,  we consider the problem of maximizing the throughput for a given budget of energy.
Formally, we   are given a set of $n$ jobs $J=\{J_1,J_2,\ldots, J_n\}$, where each job $J_j$ is characterized by its processing requirement (work) $p_j$, its release date $r_j$ and its deadline $d_j$. We consider integer release dates, deadlines and processing requirements.
(For simplicity, we suppose that the earliest released job is released at $t=0$.)  
We assume that the jobs have to be executed by a single speed-scalable processor, i.e. a processor which can vary its speed over time (at a given time, the processor's speed can be any non-negative value).
The processor can execute at most one job at each time.
We measure the processor's speed in units of executed work per unit of time.
If $s(t)$ denotes the  speed of the processor at time $t$,
then the total amount of work executed by the processor during an interval of time $[t,t')$ is equal to $\int_t^{t'}s(u)du$.
Moreover, we assume that the processor's power consumption is a convex function of its speed. 
Specifically, at any time $t$, the power consumption of the processor is $P(t)=s(t)^{\alpha}$, where $\alpha>1$ is a constant.
Since the power is defined as the rate of change of the energy consumption, the total energy consumption of the processor during an interval $[t,t')$ is $\int_t^{t'}s(u)^{\alpha}du$.
Note that if the processor runs at a constant speed $s$ during an interval of time $[t,t')$, then it executes $(t'-t)\cdot s$ units of work and it consumes $(t'-t)\cdot s^{\alpha}$ units of energy.

Each job $J_j$ can start being executed after or at its release date $r_j$.
Moreover, depending on the case, we may allow or not  the preemption of jobs, i.e. the execution of a job may be suspended and continued later from the point of suspension.
Given a budget of energy $E$, our objective is to find a schedule of maximum throughput whose energy does not exceed the budget $E$, where the throughput of a schedule is defined as the number of jobs which are completed on time, i.e. before their deadline.
Observe that a job is completed on time if it is entirely executed during the interval $[r_j,d_j)$.
By extending the well-known 3-field notation by Graham et al., this problem can be denoted as $1|pmtn,r_j,E|\sum U_j$. We also consider the weighted version of the problem where every job $j$ is also associated with a weight $w_j$ and the objective is no more the maximization of the cardinality of the jobs that are completed on time, but the maximization of the sum of their weights.  We denote this problem as $1|pmtn,r_j,E|\sum w_jU_j$. 
In what follows, we consider the problem in the case where 
all jobs have arbitrary integer release dates and deadlines.

\begin{figure}
\label{tab_summary}
\begin{center}

\begin{tabular}{|c|c|c|c|}
\hline 
Problem & Weighed profit & Assumption & Time complexity \\ 
\hline 
\multirow{3}{*}{Classical}& \multirow{3}{*}{no} &  $r_j,pmtn,p_j=p$ & $O(n\log n)$ \cite{LAWLER94}
\\
\cline{3-4}
&&\multirow{2}{*}{ $r_j,pmtn$} & $O(n^5)$ \cite{LAWLER90} \\
&& &  $O(n^4)$\cite{B99}\\

\hline

\hline
\multirow{4}{*}{Classical} &\multirow{4}{*}{yes} & \multirow{2}{*}{$r_j,pmtn,p_j=p$} & $O(n^{10})$ \cite{B99w} \\
& & & $O(n^4)$ \cite{BCDJV04}\\
\cline{3-4}
& & $r_j,pmtn$ & $O(n^3W^2)$ \cite{LAWLER90}\\ 
\cline{3-4}
&& $r_j,p_j=p$ & $O(n^{7})$ \cite{B99w}\\
\hline
\hline
\multirow{4}{*}{Speed-scaling}& \multirow{4}{*}{no} & $r_j\le r_i \Leftrightarrow d_j\le d_i$ & $O(n^6\log n\log P)$ \cite{ABCL13}\\
\cline{3-4}
&& $r=0$ & $O(n^4\log n\log P)$\cite{ABCL13} \\ 
\cline{3-4}
&& $r_j,p_j=p$& $O(n^{21})$ \textbf{[this paper]}\\
\cline{3-4}
&& $r_j,pmtn$& $O(n^{6}L^9P^9)$ \textbf{[this paper]}\\
\hline
\multirow{4}{*}{Speed-scaling}& \multirow{4}{*}{yes} & $r_j\le r_i \Leftrightarrow d_j\le d_i$ & $O(n^4W^2\log n\log P)$ \cite{ABCL13}\\
\cline{3-4}
&& $r=0$ & $O(n^2W^2\log n\log P)$\cite{ABCL13} \\ 
\cline{3-4}
&& $r_j,p_j=p$& $O(n^{19}W^2)$ \textbf{[this paper]}\\
\cline{3-4}
&& $r_j,pmtn$& $O(n^{2}W^4L^9P^9)$ \textbf{[this paper]}\\
\hline

\end{tabular} 
\caption{Summary of offline Throughput maximization}

\end{center}
\end{figure}

\subsection{Related Works and our Contribution}
Angel et al. \cite{ABCL13} were the first to consider the
throughput maximization problem in the energy setting for the offline case. They studied the problem for a  
particular family of instances where the jobs have agreeable deadlines, i.e. for every pair of jobs $J_i$ and $J_j$, $r_i \leq r_j$ if and only if $d_i \leq d_j$. 
They provided
a polynomial time algorithm to solve the problem for agreeable instances. However, to the best of our knowledge, the complexity of the unweighted preemptive problem for arbitrary instances remained unknown until now. In this paper, we prove that there is a pseudo-polynomial time algorithm for solving the problem optimally. For the weighted version, the problem is $\mathcal{NP}$-hard even for instances in which all the jobs have common release dates and deadlines\cite{ABCL13}. Angel et al. showed that the problem admits a pseudo-polynomial time algorithm for agreeable instances. Our algorithm for the unweighted case can be adapted for the weighted throughput problem with arbitrary release dates and deadlines solving the problem in pseudo-polynomial time.
More recently, Antoniadis et al. \cite{MFCS13} considered a generalization of the classical knapsack problem where the objective is to maximize the total profit of the chosen items minus the cost incurred by their total weight. The case where the cost functions are convex can be translated in terms of a weighted throughput problem where the objective is to select the most profitable set of jobs taking into account the energy costs. Antoniadis et al. presented a FPTAS and a fast 2-approximation algorithm for the non-preemptive problem where the jobs have no release dates or deadlines.
We also consider the non-preemptive case in the special case where all jobs have equal processing requirements. For the unweighted version, we propose a strongly polynomial-time algorithm that gives an optimal solution. This result answers an open question left in \cite{BKLLN13} concerning the complexity of the non-preemptive energy-minimization problem when all the jobs have equal processing requirements. For the weighted case, our algorithm can solve the problem in pseudo-polynomial time.  
In Figure~\ref{tab_summary}, we summarize the results in the offline context for throughput maximization in the classical scheduling setting (with no energy considerations) as well as in the speed scaling context.


Different variants of throughput maximization in the speed scaling context have been studied in the literature (see \cite{CCLLMW07,BCLL08,LLTW07,CLMW07,Li11,CLL10}), but in what follows we focus on the complexity status of the offline case.

The paper is organized as follows: we first present an optimal algorithm
for the general case of throughput maximization where preemption
is allowed.
Finally, we consider the case where all the jobs have the 
same processing requirement and preemption is not allowed.

Because of space limitations some proofs are omitted.

\section{Preliminaries}

Among the schedules of maximum throughput, we try to find the one of minimum energy consumption.
Therefore, if we knew by an oracle the set of jobs $J^*$, $J^*\subseteq J$, which are completed on time in an optimal solution, we would simply have to apply an optimal algorithm for $1|pmtn,r_j,d_j|E$ for the jobs in $J^*$ in order to determine a minimum energy schedule of maximum throughput for our problem. Such
an algorithm has been proposed in \cite{YDS95}.
Based on this observation, we can use in our analysis some properties of an optimal schedule for $1|pmtn,r_j,d_j|E$.

Let $t_1,t_2,\ldots,t_k$ be the time points which correspond to release dates and deadlines of the jobs so that for each release date and deadline there is a $t_i$ value that corresponds to it.
We number the $t_i$ values in increasing order, i.e. $t_1<t_2<\ldots<t_k$. The following theorem is a consequence of
the algorithm \cite{YDS95} and was proved in \cite{BKP07}.

\begin{theorem} \label{theorem:OptimalPropertiesYao}
A feasible schedule for $1|pmtn,r_j,d_j|E$ is optimal if and only if all the following hold:
\begin{enumerate}
\setlength{\itemsep}{0pt}
	\item Each job $J_j$ is executed at a constant speed $s_j$.
	\item The processor is not idle at any time $t$ such that $t\in(r_j,d_j]$, for all 			$J_j\in J$.
	\item The processor runs at a constant speed during any interval $(t_i,t_{i+1}]$, for $1\leq i\leq k-1$.
	\item If others jobs are scheduled in the span $[r_j,d_j]$ of $J_j$, then their speed is necessary greater or equal to the speed of $J_j$.
\end{enumerate}
\end{theorem}

Theorem \ref{theorem:OptimalPropertiesYao} is also 
satisfied by the optimal schedule of
$1|pmtn,r_j,E|\sum U_j$ for the jobs in $J^*$.
We suppose that jobs are sorted in non-decreasing order
of their deadline, i.e. 
$d_1\le d_2 \le \ldots \le d_n$.
Moreover, we suppose that release dates, deadlines and processing 
requirements are integer.

\begin{definition}
Let $J(k,s,t)=\{ J_j | j\leq k \mbox{ and } s\le r_j < t \}$ 
be the set of the first $k$ jobs
according to the earliest deadline first ({\sc edf}) order and such that their release date are
within $s$ and $t$.
\end{definition}

\begin{lemma}\label{lemma_critical}
The total period in which the processor runs at a same speed
in an optimal solution for $1|pmtn,r_j,d_j|E$ has an integer length.
\end{lemma}
\begin{proof}
The total period is defined by a set of intervals $(t_i,t_{i+1}]$
for $1\leq i\leq k-1$ thanks to the property 3) in Theorem~\ref{theorem:OptimalPropertiesYao}.
Since each $t_i$ corresponds to some release date or some deadline,
then $t_i\in \mathbb{N},~1\le i\le k$.
Thus every such period has necessarily an integer length.
\end{proof}

\begin{definition}
Let $L=d_{max} - r_{min} $ be the span of the whole schedule. 
To simplify the notation, we assume that $r_{min}=0$.
\end{definition}

\begin{definition}
Let $P=\sum_j p_j$ be the total processing requirement of all the jobs.
\end{definition}

\begin{definition}
We call an {\sc edf} schedule, a schedule such that
at any time,
the processor schedules the job that has the smallest
deadline among the set of available jobs
at this time.
\end{definition}
In the sequel, all the schedules considered are {\sc edf} schedules.

\section{The preemptive case}

In this part, we propose an optimal algorithm which
is based on dynamic programming depending on the span
length $L$ and the maximum processing requirement $P$. 
As mentioned previously,
among the schedules 
of maximum throughput, our algorithm constructs a schedule
with the minimum energy consumption.

For a subset of jobs $S\subseteq J$, a schedule which involves only the jobs 
in $S$ will be called a $S$-schedule.

\begin{definition}
Let $G_k(s,t,u)$ be the minimum energy consumption
of a $S$-schedule with $S\subseteq J(k,s,t)$ such that 
$|S|=u$ and such that the jobs in $S$ are 
entirely scheduled in $[s,t]$ 
\end{definition}

Given a budget $E$ that we cannot exceed, the function objective is
$\max\{u~|~G_n({0,d_{max},u}) \le E;~0\le u\le n \}$.

\begin{definition}
Let $F_{k-1}(x,y,u,\ell,i,a,h)$ be the minimum energy consumption
of a $S$-schedule with $S\subseteq J(k-1,s,t)$ such that 
$|S|=u$ and such that the jobs in $S$ are 
entirely scheduled in $[x,y]$ during at most $a+h\times\frac{l}{i}$ unit times. 
Moreover, we assume that each maximal block of continuous jobs
of $\mathcal{S}$ starts at a release date
and has a length equal to $a'+h'\times\frac{\ell}{i}$
with $a',h'\in \mathbb{N}$.
\end{definition}

Next, we define the set of all important time of an optimal schedule
in which every job can starts and ends, and we show that the size
of this set is pseudo-polynomial.

\begin{definition}\label{def:Omega}
Let $\Omega= \{r_j~|~ j=1,\ldots, n\} \cup \{d_j~|~ j=1,\ldots, n\}$.
\end{definition}

\begin{definition}\label{def:Phi}
Let $\Phi =\{s+ h\times \frac{\ell}{i}\le L~|~ i = 1,\ldots ,P;~ 
 h = 0,\ldots ,i;~ s=0,\ldots, L;~ \ell = 1,\ldots, L
 \}$
\end{definition}

\begin{prop}\label{prop_phi}
There exists an optimal schedule $\mathcal{O}$ in which for each job, its
starting times and finish times belong to the set $\Phi$, and such that each 
job is executed with a speed $\frac{i}{\ell}$ for some $i=1,\ldots,P$ and 
$\ell=1,\ldots,L$.
\end{prop}
\begin{proof}
W.l.o.g. we can consider that each job has an unit processing requirement. If 
it is not the case, we can split a job $J_j$ into $p_j$ jobs, each one with
an unit processing requirement.

We briefly explain the algorithm in~\cite{YDS95} which gives an optimal 
schedule. At each step, it selects the (critical) interval I=[s,t] with
$s$ and $t>s$ in $\Omega= \{r_j~|~ j=1,\ldots, n\} \cup 
 \{d_j~|~ j=1,\ldots, n\}$, such that 
$s_I=\frac{|\{J_j~|~s\le r_j\le d_j\le t\}|}{t-s}$ is maximum.
All the jobs inside this interval are executed at the speed $s_I$, which is
of the form $\frac{i}{\ell}$ for some $i=1,\ldots,P$ and $\ell=1,\ldots,L$,
and according to the {\sc edf} order.
This interval cannot be used any more, and we recompute a new critical interval
without considering the jobs and the previous critical intervals, until all
the jobs have been scheduled.

We can remark that the length of each critical interval (at each step) 
$I=[s,t]$ is an integer. This follows from the fact that $s=r_i\in \mathbb{N}$ 
for some job $J_i$, and $t=d_j\in \mathbb{N}$ for some job $J_j$, moreover we
remove integer lengths at each step (the length of previous critical
intervals which intersect the current one), so the new considered critical 
interval has always an integer length.

Then we can define every start times or completion times of each job in this 
interval.
We first prove that the end time of a job in a continuous critical interval,
i.e. a critical interval which has an empty intersection with all other
critical intervals, belongs to $\Phi$. Let $J_k$ be any job in a continuous 
critical interval and let $k_s$ and $k_e$ be respectively its start and end 
times.  Then there is no idle time between $s=r_x$ (for some $J_x$)
and $k_e$ since it is a critical interval. Let $v=\frac{i}{\ell}$ be the 
processor speed in this interval and $p=\frac{\ell}{i}$ be the processing time 
of a job.  The jobs that execute (even partially) between $k_s$ and $k_e$
execute neither before $k_s$ nor after $k_e$ since it is an
{\sc edf} schedule. Thus $k_e-k_s$ is a multiple of $p$.
Two cases can occur:
\begin{itemize}
\setlength{\itemsep}{0pt}
\item Either $J_k$ causes an interruption and hence $k_s = r_k$.
\item Or $J_k$ does not cause any interruption and hence the jobs that execute
between $s$ and $k_s$, are fully scheduled in this interval. Consequently, 
$k_s-s$ is a multiple of $p$.
\end{itemize}
\vspace*{-5pt}
In both cases, there is a release date $r_y$ (either $r_k$ or $r_x$) such that 
between $r_y$ and $k_e$, the processor is never idle and such that $k_e$ is 
equal to $r_y$ modulo $p$. On top of that, the distance between $r_y$ and $t$
is not greater than $n\times p$.  Hence, $k_e\in \Phi$.
Now consider the start time of any job. This time point is either the
release date of the job or is equal to the end time of the "previous" one. 
Thus, starting times also belong to $\Phi$.

Now we consider the start and the completion times of a job in a 
critical interval $I$ in which there is at least one another critical interval 
(with greater speeds) included in $I$ or intersecting $I$. Let $A$ be the 
union of those critical intervals.
Since jobs of $I$ cannot be scheduled on the intervals $A$, the start time and
completion time of those jobs have to be (right)-shifted with an integer value
(since each previously critical interval has an integer length).
Thus the starts time and completion time for all jobs still belong to $\Phi$.
\end{proof}

\begin{prop}\label{prop_G}
One has
\[  G_k(s,t,u)=\min \left\{
\begin{aligned}
& G_{k-1}(s,t,u)&\\
& \min_{ 
\substack{x\in \Phi\\
0\le u_1 \le u\\
0\le u_2 \le u\\
0\le u_1+u_2\le u-1\\
0\le a\le L;~ 
1\le \ell \le L\\
1\le i \le P;~
0\le h \le P\\
{y-x=a+(p_k+h)\frac{\ell}{i}}\\
r_k\le x \le y \le d_k 
}}&
 \left\{
\begin{array}{l}
 G_{k-1}(s,x,u_1)
+F_{k-1}(x,y,u_2,\ell,i,a,h)\\
+\Big(\dfrac{i}{\ell}\Big)^{\alpha -1}p_k
+ G_{k-1}(y,t,u-u_1-u_2-1)
\end{array}
\right\}
\end{aligned}
\right.
\]
$G_0(s,t,0)=0~\forall s,t\in \Phi$\\
$G_0(s,t,u)=+\infty~\forall s,t\in \Phi$ and $u>0$
\end{prop}

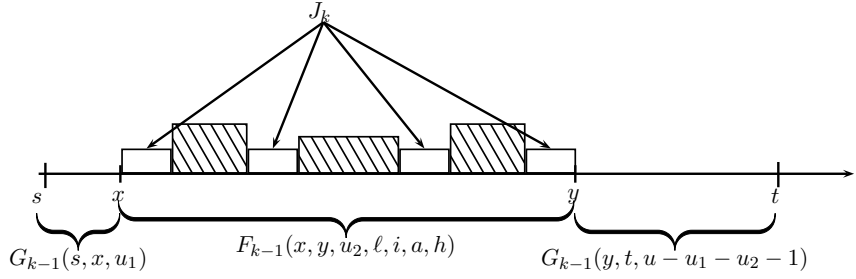
\begin{figure}[!h]
\begin{center}
\scalebox{0.8} 
{
\begin{pspicture}(0,-2.978125)(19.655468,2.978125)
\psline[linewidth=0.04cm,arrowsize=0.05291667cm 2.0,arrowlength=1.4,arrowinset=0.4]{->}(0.27546877,-1.0603125)(19.63547,-1.0603125)
\usefont{T1}{ppl}{m}{n}
\rput(0.22921875,-1.64375){$s$}
\usefont{T1}{ppl}{m}{n}
\rput(2.1692188,-1.58375){$x$}
\usefont{T1}{ppl}{m}{n}
\rput(17.759218,-1.58375){$t$}
\psline[linewidth=0.04cm](2.2154686,-0.85662806)(2.2154686,-1.2717886)
\usefont{T1}{ppl}{m}{n}
\rput(12.9792185,-1.58375){$y$}
\psline[linewidth=0.04cm](13.015469,-0.85662806)(13.015469,-1.2717886)
\psframe[linewidth=0.04,dimen=outer,fillstyle=vlines,hatchwidth=0.04,hatchangle=30.0](5.235469,0.1396875)(3.4354687,-1.0603125)
\psframe[linewidth=0.04,dimen=outer,fillstyle=vlines,hatchwidth=0.04,hatchangle=30.0](11.835468,0.1396875)(10.035469,-1.0603125)
\usefont{T1}{ptm}{m}{n}
\psline[linewidth=0.04cm](0.43546876,-0.8003125)(0.43546876,-1.4003125)
\psline[linewidth=0.04cm](17.835468,-0.7803125)(17.835468,-1.3803124)
\psframe[linewidth=0.04,dimen=outer,fillstyle=vlines,hatchwidth=0.04,hatchangle=30.0](8.835468,-0.1603125)(6.4354687,-1.0603125)
\usefont{T1}{ptm}{m}{n}
\psbrace[rot=90,ref=C,nodesepB=4pt](0.23,-2)(2.2,-2){$G_{k-1}(s,x,u_1)$}
\psbrace[rot=90,ref=C,nodesepB=4pt](13,-2)(17.76,-2){$G_{k-1}(y,t,u-u_1-u_2-1)$}
\psbrace[rot=90,ref=C,nodesepB=4pt](2.17,-1.7)(12.98,-1.7){$F_{k-1}(x,y,u_2,\ell,i,a,h)$}
\usefont{T1}{ptm}{m}{n}
\psframe[linewidth=0.04,dimen=outer](3.4354687,-0.4603125)(2.2354689,-1.0603125)
\usefont{T1}{ptm}{m}{n}
\rput(6.946875,2.7896874){$J_k$}
\psframe[linewidth=0.04,dimen=outer](6.4354687,-0.4603125)(5.235469,-1.0603125)
\psframe[linewidth=0.04,dimen=outer](13.035469,-0.4603125)(11.835468,-1.0603125)
\psframe[linewidth=0.04,dimen=outer](10.035469,-0.4603125)(8.835468,-1.0603125)
\psline[linewidth=0.04cm,arrowsize=0.05291667cm 2.0,arrowlength=1.4,arrowinset=0.4]{->}(7.0354686,2.5396874)(2.8354688,-0.4603125)
\psline[linewidth=0.04cm,arrowsize=0.05291667cm 2.0,arrowlength=1.4,arrowinset=0.4]{->}(7.0354686,2.5396874)(5.835469,-0.4603125)
\psline[linewidth=0.04cm,arrowsize=0.05291667cm 2.0,arrowlength=1.4,arrowinset=0.4]{->}(7.0354686,2.5396874)(9.435469,-0.4603125)
\psline[linewidth=0.04cm,arrowsize=0.05291667cm 2.0,arrowlength=1.4,arrowinset=0.4]{->}(7.0354686,2.5396874)(12.435469,-0.4603125)
\end{pspicture} 
}

\end{center}
\caption{Illustration of Proposition~\ref{prop_G}}
\end{figure}

\begin{proof}
Let $G'$ be the right hand side of the formula, 
$G'_1$ be the first line of $G'$ and 
$G'_2$ be the second line of $G'$.

\textbf{We first prove that $G_k(s,t,u)\le G'$.}


Since $J(k-1,s,t)\subseteq J(k,s,t)$,
then $ G_{k}(s,t,u)\le G_{k-1}(s,t,u)= G'_1$.

Now consider a schedule 
$\mathcal{S}_1$ that realizes $G_{k-1}(s,x,u_1)$, 
$\mathcal{S}_2$ that realizes $F_{k-1}(x,y,u_2,\ell,i,a,h)$ 
such that $y-x=a+(p_k+h)\times\frac{\ell}{i}$
and
$\mathcal{S}_3$ that realizes $G_{k-1}(y,t,u-u_1-u_2-1)$.
We build a schedule with $\mathcal{S}_1$ from $s$ to $x$, with 
$\mathcal{S}_2$ from $x$ to $y$ and with 
$\mathcal{S}_3$ from $y$ to $t$.

Since 
$F_{k-1}(x,y,u_2,\ell,i,a,h)$ is a schedule where the processor
execute jobs during at most $a+h\times\frac{l}{i}$ unit times
and we have $y-x=a+(p_k+h)\times\frac{\ell}{i}$, then 
there is at least $p_k\times\frac{\ell}{i}$ units time for $J_k$.
Thus $J_k$ can be
scheduled with speed $\frac{i}{\ell}$ during $[x,y]$.

Obviously, the subsets $J(k-1,s,x), J(k-1,x,y)$ 
and $J(k-1,y,t)$ do not intersect, so
this is a feasible schedule, and its cost is $G'_2$,
thus $G_{k}(s,t,u)\le G'_2 $.\\

\textbf{We now prove that $G'\le G_k(s,t,u)$.}

If $J_k \notin \mathcal{O}$ such that $\mathcal{O}$ realizes $G_k(s,t,u)$,
then $G'_1= G_k(s,t,u)$.

Now, let us consider the case $J_k \in \mathcal{O}$.

We denote by $\mathcal{X}$ the schedule that realizes $G_k(s,t,u)$ 
in which
the first starting time $x$ of $J_k$ is maximal,
and $y$ the last completion time of $J_k$ is
also maximal. According to Proposition~\ref{prop_phi}, we
assume that $x,y\in \Phi$.
We split $\mathcal{X}$ into three sub-schedules 
$\mathcal{S}_1\subseteq J(k-1,s,x)$,
$\mathcal{S}_2\subseteq J(k-1,x,y)\cup \{J_k\}$ and 
$\mathcal{S}_3\subseteq J(k-1,y,t)$.

We claim that we have the following properties:
\begin{itemize}
\setlength{\itemsep}{-2pt}
\item[P1)] all the jobs of $\mathcal{S}_1$ are released in $[s,x]$ 
and are completed before $x$,
\item[P2)] all the jobs of $\mathcal{S}_2$ are released in $[x,y]$ 
and are completed before $y$,
\item[P3)] all the jobs of $\mathcal{S}_3$ are released in $[y,t]$ 
and are completed before $t$.
\end{itemize}

\textit{We prove P1.\\}
Suppose that there is a job $J_j\in \mathcal{S}_1$ which is not
completed before $x$. Then we can swap some part 
of $J_j$ of length $\ell$ which is scheduled after $x$
with some part of $J_k$ of length $\ell$ at time $x$. This can be done
since we have $d_j\le d_k$. Thus we
have a contradiction with the fact that $x$ was maximal.\\

\textit{We prove P2.\\}
Similarly, suppose that there is a job 
$J_j\in \mathcal{S}_2$ which is not
completed before $y$. Then we can swap some part 
of $J_j$ of length $\ell$ which is scheduled after $y$
with some part of $J_k$ of length $\ell$ in $[x,y]$. It can be done
since we have $d_j\le d_k$. Thus we
have a contradiction with the fact that $y$ was maximal.\\

\vspace*{-2pt}
\textit{We prove P3.\\}
If there exists a job in $\mathcal{S}_3$ which is not completed
at time $t$, then the removal of this job would lead
to a lower energy consumption schedule for
$\mathcal{S}_3$ which contradicts the definition of
$G_{k-1}(y,t,|\mathcal{S}_3|)$\\

Let us consider the schedule 
$\mathcal{S}'_2=\mathcal{S}_2\setminus J_k$ in $[x,y]$.
Since $[x,y]\subseteq [r_k,d_k]$, thanks to property 4) of 
Theorem~\ref{theorem:OptimalPropertiesYao}, the speeds of jobs
in $\mathcal{S}'_2$ are necessarily greater or equal to the speed
of $J_k$. 
Let us consider any maximal block $b$ of continuous jobs in 
$\mathcal{S}'_2$.
This block can be partitioned into two sub-blocks $b_1$ and $b_2$ such that $b_1$ (resp. $b_2$) contains all the jobs
of $b$ which are scheduled with a speed equals to (resp. strictly greater than) the speed of $J_k$.
All the jobs scheduled in block $b$ are also totally completed in $b$ (this comes from the {\sc edf} property and because
$J_k$ has the biggest deadline).
Notice that the speed of $J_k$ is equal to $\frac{i}{\ell}$
for some value $i=1,\ldots,P$ and $\ell=1,\ldots,L$ thanks to
Proposition~\ref{prop_phi}.
Thus the total processing time of $b_1$ is necessarily
$h'\times\frac{\ell}{i}$.
Moreover since from property 3) of Theorem~\ref{theorem:OptimalPropertiesYao}, all the speed changes occur at time $t_i\in\mathbb{N}$,
the block $b_2$ has an integer length.
Therefore, every block $b$ has a length equal to $a'+h'\times\frac{\ell}{i}$ and the total
processing time of $\mathcal{S}'_2$
is $a+h\times\frac{\ell}{i}$.
Moreover, every block $b$ in $\mathcal{S}'_2$ starts at a release date (this comes from the {\sc edf} property).
On top of that, we have $y-x=a+(h+p_k)\times\frac{\ell}{i}$ with 
$a=0,1,\ldots, L$ and $h=0,\ldots,i$.
Moreover, every block $b$ in $\mathcal{S}'_2$ starts at a release date (this comes from the {\sc edf} property).
Hence the cost of the schedule $\mathcal{S}'_2$ is greater than 
$F_{k-1}(x,y,|\mathcal{S}'_2|,\ell,i,a,h)$.
The energy consumption of $J_k$ is exactly $p_k\times(\frac{i}{\ell})^{\alpha -1}$.

Similarly, the cost of the schedule 
$\mathcal{S}_1$ is greater than $G_{k-1}(s,x,|\mathcal{S}_1|)$ and
the cost of $\mathcal{S}_3$ is greater than 
$G_{k-1}(y,t,|\mathcal{S}_3|)$.

Therefore, 
$ G_k(s,t,u)\ge  
G_{k-1}(s,x,|\mathcal{S}_1|)  +
F_{k-1}(x,y,|\mathcal{S}_2|,\ell,i,a,h) +
G_{k-1}(y,t,|\mathcal{S}_3|)
+p_k\Big(\dfrac{i}{\ell}\Big)^{\alpha -1}
= G'_2$
and
$ G_k(s,t,u)\ge G'$.
\end{proof}

\begin{prop}\label{prop_F}
One has
\[  F_{k-1}(x,y,u,\ell,i,a,h)=
\min_{
\substack{
0\le a'\le a;~0\le h'\le h\\
x\le x'=r_j\le y;~j\le k\\
1\le \beta\le u\\
y'=x'+a'+h'\times\frac{\ell}{i}\le y
}} 
\{G_{k-1}(x',y',\beta)+F_{k-1}(y',y,u-\beta,\ell,i,a-a',h-h')\}\\
\]

$F_{k-1}(x,y,0,\ell,i,a,h)=0$

$F_{k-1}(x,y,u,\ell,i,0,0)=+\infty$\\
\end{prop}

\begin{prop}\label{prop_complexity_preemp}
The preemptive case problem can be solved in $O(n^6L^9P^9)$ time and 
in $O(nL^6P^6)$ space.
\end{prop}

The dynamic program can be adapted for the 
weighted version and 
has a running time of $O(n^2W^4L^9P^9)$
where $W$ is the sum of the weight of all jobs.

\section{The non-preemptive case}


\begin{definition}\label{def:Eksxtu}
For $1\leq u\leq |J(k,s,t)|$, we define $E_k(s,x,t,u)$ as the minimum energy consumption 
of an $S$-schedule such that $|S|=u$, $S\subseteq J(k,s,t)$ where the jobs scheduled 
cannot be preempted such that all jobs in 
$S$ are executed entirely in $[s,t]$
and such that 
the processor is idle in $[s,x]$

If such a schedule does not exist, i.e. when $u>|J(k)|$, 
then $E_k(s,x,t,u)=+\infty$.
\end{definition}

Given a budget $E$, the function objective is
$\max\{u~|~E_n({0,0,d_{max},u}) \le E \}$.

\begin{definition}\label{def:Theta_s_t}
Let $\Theta_{s,t} =\{s+ h\times \dfrac{t-s}{i}~|~ i = 1,\ldots ,n \mbox{ and } 
 h = 0,\ldots ,i
 \}$  and $\Theta=\bigcup_{s,t}\{\Theta_{s,t}~|~s,t\in\Omega\}$.
\end{definition}

\begin{prop}\label{theta}
Let $\mathcal{O}$ be any optimal schedule, then for each job in $\mathcal{O}$
its starting time and completion time belong to the set $\Theta$.
\end{prop}
\begin{proof}
This optimal schedule $\mathcal{O}$ can be decomposed into successive block
of tasks, $B_1, \ldots, B_i, \ldots$, such that there is no idle time inside 
a block, and such that there is idle time immediately before and after each
block. Let us consider any such block $B$, and let $J_j$ (resp. $J_{j'}$) the
first (resp. last) task executed in $B$. The starting time of $J_{j}$ is 
necessarily $r_j$, otherwise the schedule would not be optimal since it would
be possible to decrease the energy consumption by starting the task $J_j$ 
earlier. A similar argument shows that the completion time of task $J_{j'}$ is
necessarily $d_{j'}$. Now the block $B$ can be decomposed into a set 
$b_1,\ldots , b_i, \ldots$ of maximal sub-blocks of consecutive jobs such 
that all the jobs executed inside a sub-block $b_i$ are scheduled with
the same common speed $s_i$. Let $J_j$ and $J_{j'}$ be two consecutive jobs such
that $J_j$ and $J_{j'}$ belong to two consecutive sub-blocks, let's say $b_i$ and 
$b_{i+1}$. Then either $s_i > s_{i+1}$ or $s_i < s_{i+1}$.
In the first case, the completion time of $J_j$ (which is also the starting time
of $J_{j'}$) is necessarily $d_j$, otherwise we could obtain a better schedule
by decreasing (resp. increasing) the speed of task $J_j$ (resp. $J_{j'}$).
For the second case, a similar argument shows that the completion time of $J_j$
is necessarily $r_{j'}$. Let us consider now a sub-block. The previous 
arguments have shown that the starting time $s$ and completion time $t$ of 
that sub-block belong to the set $\Omega$.
Since $p_j=p$~$\forall j$, all the jobs scheduled inside this sub-block 
have their starting time and completion time that belong to $\Theta_{s,t}$.
\end{proof}


\begin{prop}\label{prop_E}
One has
\[  E_k(s,x,t,u)=\min \left\{
\begin{aligned}
& E_{k-1}(s,x,t,u)\\
& \min_{
\substack{
s'\in \Theta;~x'\in \Theta\\
s \le r_k \le  s' < x' \le t\\
x'\le d_k\\
0\le \ell \le u-1}
}
\left\{
 E_{k-1}(s,x,s',\ell)  +
 E_{k-1}(s',x',t,u-\ell-1) + \dfrac{p^\alpha}{(x'-s')^{\alpha-1}}  
\right\}
\end{aligned}
\right.
\]

$E_0(s,x,t,0)=0$

$E_0(s,x,t,u)=+ \infty~\forall u>0$ \\

\end{prop}

\begin{figure}[H]
\centering
\scalebox{1} 
{
\begin{pspicture}(0,-2)(13.22,1.1)
\psline[linewidth=0.04cm,arrowsize=0.05291667cm 2.0,arrowlength=1.4,arrowinset=0.4]{->}(0.0,-0.32437494)(13.2,-0.32437494)
\usefont{T1}{ppl}{m}{n}
\rput(0.43375,-0.7543749){$s$}
\usefont{T1}{ppl}{m}{n}
\rput(2.79375,-0.83437496){$x$}
\usefont{T1}{ppl}{m}{n}
\rput(12.38375,-0.83437496){$t$}
\usefont{T1}{ppl}{m}{n}
\usefont{T1}{ppl}{m}{n}
\psline[linewidth=0.04cm](2.84,-0.10725305)(2.84,-0.52241355)
\usefont{T1}{ppl}{m}{n}
\rput(5.8337502,-0.8143749){$s'$}
\usefont{T1}{ppl}{m}{n}
\rput(8.233749,-0.83437496){$x'$}
\psline[linewidth=0.04cm](8.24,-0.10725305)(8.24,-0.52241355)
\psframe[linewidth=0.04,dimen=outer,fillstyle=vlines,hatchwidth=0.04,hatchangle=30.0,hatchsep=0.20120001](8.24,0.27906257)(5.84,-0.32093745)
\psframe[linewidth=0.04,dimen=outer,fillstyle=vlines,hatchwidth=0.04,hatchangle=30.0,hatchsep=0.20120001](2.84,0.27906257)(0.44,-0.32093745)
\psline[linewidth=0.04cm](0.44,-0.10725305)(0.44,-0.52241355)
\psline[linewidth=0.04cm](5.84,-0.10725305)(5.84,-0.52241355)
\psline[linewidth=0.04cm](12.44,-0.10725305)(12.44,-0.52241355)
\usefont{T1}{ptm}{m}{n}

\rput(9.5314064,1.2){$J_k$}
\psline[linewidth=0.04cm,arrowsize=0.05291667cm 2.0,arrowlength=1.4,arrowinset=0.4]{->}(9.1,1.1)(7.1,0.3690625)

\psbrace[rot=90,ref=C,nodesepB=4pt](0.4,-1)(5.83,-1){$E_{k-1}(s,x,s',\ell) $}
\psbrace[rot=90,ref=C,nodesepB=4pt](5.84,-1)(12.44,-1){$E_{k-1}(s',x',t,u-\ell-1)$}

\end{pspicture} 
}

\caption{Illustration of Proposition~\ref{prop_E}}
\end{figure}
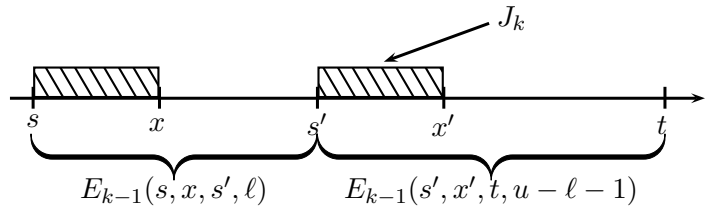

\begin{proof}

Let $E'$ be the right hand side. We distinguish by $E'_1$ for the first line
of $E'$ and $E'_2$ for the second line.\\

\textbf{We first prove that $E_k(s,x,t,u)\le E'$.}

Since $J(k-1,s,t)\subseteq J(k,s,t)$,
then $ E_{k}(s,t,u,f)\le E_{k-1}(s,t,u,f)= E'_1$.

Now consider a schedule $\mathcal{S}_1$ 
that realizes $E_{k-1}(s,x,s',\ell)$ and $\mathcal{S}_2$ 
that realizes $E_{k-1}(s',x',t,u-\ell-1)$.
We build a schedule with $\mathcal{S}_1$ from $s$ to $s'$ and with 
$\mathcal{S}_2$ from $s'$ to $t$. Job $J_k$ is scheduled within $\mathcal{S}_2$
during $[s',x']$.
Obviously, the subsets $J(k-1,s,s')$ and $J(k-1,s',t)$ do not intersect, so
this is a feasible schedule which cost is exactly $E'_2$,
thus $E_k(s,x,t,u)= E'_2 \le E'$.\\

\textbf{We now prove that $E'\le E_k(s,x,t,u)$.}

If $J_k \notin \mathcal{O}$ such that $\mathcal{O}$ realizes $E_k(s,x,t,u)$,
then $E'\le E'_1= E_k(s,x,t,u)$.

Now, let consider the case where $s'\in\Theta$
and some value $\ell$.
The two subsets of jobs do not intersect.
We denote $\mathcal{X}$ the schedule that realizes $E_k(s,x,t,u)$ in which
the starting time of $J_k$ is maximal, i.e. $s'$ is maximal.
We split into two sub-schedule $\mathcal{S}_1\subseteq \mathcal{X}$ and $\mathcal{S}_2\subseteq \mathcal{X}$.
We claim that all the jobs of $\mathcal{S}_1$ are released in $[s,s']$ 
and are completed before $s'$.
If it is not the case, 
it means that there exists a job which has been preempted and it is not a feasible solution.
We claim that on $\mathcal{X}$, the 
jobs executed after $J_k$ are not available when it starts (i.e. their release date
is strictly greater than the start time of $J_k$). Indeed,
suppose that there is a job $J_j$ that starts after $J_k$ and that is available 
at the start time of $J_k$. 
Then we can swap the two jobs in the schedule
such that
jobs swap also their speed, i.e.
the length are preserved (see Figure~\ref{fig_swap_E}).

\begin{figure}[H]
\label{fig_swap_E}
\begin{center}
\scalebox{0.8} 
{
\begin{pspicture}(0,-2.2)(17.800938,2.3176563)
\psline[linewidth=0.04cm,arrowsize=0.05291667cm 2.0,arrowlength=1.4,arrowinset=0.4]{->}(0.5209375,1.1042187)(17.780937,1.1176562)
\usefont{T1}{ppl}{m}{n}
\rput(0.8946875,0.27421877){$s$}
\usefont{T1}{ppl}{m}{n}
\rput(3.3146875,0.59421873){$x$}
\usefont{T1}{ppl}{m}{n}
\rput(15.904687,0.59421873){$t$}
\psline[linewidth=0.04cm](3.3609376,1.3213407)(3.3609376,0.90618014)
\usefont{T1}{ppl}{m}{n}
\rput(6.274688,0.2542188){$s'$}
\psline[linewidth=0.04cm](6.3809376,1.1176562)(6.3609376,0.5076563)
\usefont{T1}{ppl}{m}{n}
\rput(8.754687,0.59421873){$x'$}
\psline[linewidth=0.04cm](8.760938,1.3213407)(8.760938,0.90618014)
\psframe[linewidth=0.04,dimen=outer](8.760938,1.7076563)(6.3609376,1.1076562)
\psline[linewidth=0.04cm](0.9609375,1.3213407)(0.9609375,0.90618014)
\psline[linewidth=0.04cm](15.9609375,1.3213407)(15.9609375,0.90618014)
\psframe[linewidth=0.04,dimen=outer](13.580937,2.3176563)(12.380938,1.1176562)
\usefont{T1}{ptm}{m}{n}
\rput(7.512344,1.3876562){$J_k$}
\usefont{T1}{ptm}{m}{n}
\rput(12.912344,1.7876562){$J_j$}
\usefont{T1}{ptm}{m}{n}
\rput(0.60234374,1.4476562){$\sigma$}
\psline[linewidth=0.04cm,arrowsize=0.05291667cm 2.0,arrowlength=1.4,arrowinset=0.4]{->}(0.5209375,-1.2957813)(17.780937,-1.2823437)
\usefont{T1}{ppl}{m}{n}
\rput(0.8946875,-2.1257813){$s$}
\usefont{T1}{ppl}{m}{n}
\rput(3.3146875,-1.8057812){$x$}
\usefont{T1}{ppl}{m}{n}
\rput(15.904687,-1.8057812){$t$}
\psline[linewidth=0.04cm](3.3609376,-1.0786593)(3.3609376,-1.4938198)
\usefont{T1}{ppl}{m}{n}
\rput(6.274688,-2.1457813){$s'$}
\psline[linewidth=0.04cm](6.3809376,-1.2823437)(6.3609376,-1.8923438)
\usefont{T1}{ppl}{m}{n}
\rput(8.754687,-1.8057812){$x'$}
\psline[linewidth=0.04cm](8.760938,-1.0786593)(8.760938,-1.4938198)
\psframe[linewidth=0.04,dimen=outer](8.760938,-0.6923437)(6.3609376,-1.2923437)
\psline[linewidth=0.04cm](0.9609375,-1.0786593)(0.9609375,-1.4938198)
\psline[linewidth=0.04cm](15.9609375,-1.0786593)(15.9609375,-1.4938198)
\psframe[linewidth=0.04,dimen=outer](13.580937,-0.0823438)(12.380938,-1.2823437)
\usefont{T1}{ptm}{m}{n}
\rput(7.472344,-1.0123438){$J_j$}
\usefont{T1}{ptm}{m}{n}
\rput(12.952344,-0.6123438){$J_k$}
\usefont{T1}{ptm}{m}{n}
\rput(0.64234376,-0.9523438){$\sigma'$}
\end{pspicture} 
}

\end{center}
\caption{Illustration of the swap argument}
\end{figure}
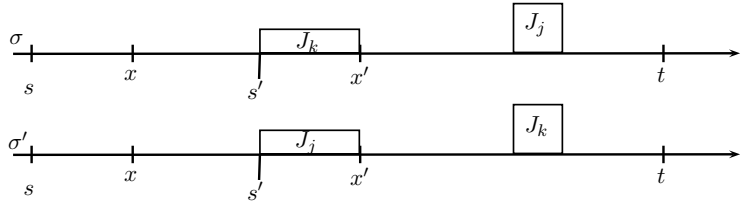

Since we have $p_j=p~\forall j$, then the 
energy consumption does not increase.
Indeed, the energy consumption of $J_j$ in the
new schedule is equal to the energy consumption
of $J_k$ in the previous schedule.
Thus we have a contradiction with the fact that $x'$ was maximal.\\
Hence the cost of the schedule 
$\mathcal{S}_1$ is greater than $E_{k-1}(s,x,s',\ell)$. Similarly,
the cost of the schedule $\mathcal{S}_2$ is greater than 
$ E_{k-1}(s',x',t,u-\ell-1) +\frac{p^\alpha}{(x'-s')^{\alpha-1}}$.\\
$ E_k(s,x,t,u)=  
  E_{k-1}(s,x,s',\ell)  +
 E_{k-1}(s',x',t,u-\ell-1) + \frac{p^\alpha}{(x'-s')^{\alpha-1}} \ge E'_2 \ge E'$
\end{proof}

\begin{definition}\label{def:Gamma_x}
Let $\Gamma(s)= \{ a+(h+1)\times\frac{b-a}{i}~|~s=a+h\times\frac{b-a}{i}
\in 
\Theta_{a,b};~a,b\in \Omega;~ i = 1,\ldots ,n;~h = 0,\ldots ,i \}$.
\end{definition}

\begin{prop}\label{prop_Gamma}
The set of value of $x$ can be reduced to $\Gamma(s)$ when $s$ is fixed
and has a size of $O(n^3)$.
\end{prop}

\begin{prop}\label{prop_complexity_non_preempt}
The dynamic program has a running time of $O(|\Gamma|^2|\Theta|^3n^3)$
or $O(n^{21})$ and $O(n^{13})$ space.
\end{prop}

The dynamic program can be adapted for the 
weighted version and 
has a running time of $O(|\Gamma|^2|\Theta|^3nW^2)$ or $O(n^{19}W^2)$
where $W$ is the sum of the weight of all jobs.

\begin{prop}\label{prop_np_hard}
The $1|r_j,p_j=p,E|\sum_j w_j U_j$ problem is $\mathcal{NP}$-hard.
\end{prop}

\section{Conclusion}

In this paper, we prove that there is a pseudo-polynomial
time algorithm for solving the problem optimally. This result is a first (partial) answer to the complexity status of the throughput maximization problem in the offline setting. Our result shows that the problem is not strongly NP-hard, but the question of whether there is a polynomial time algorithm for it remains a challenging open question.

\bibliographystyle{unsrt}
\bibliography{biblio}

\newpage

\centerline{\bf \Huge Appendix}

\section*{Proof of Proposition~\ref{prop_F}}

\begin{proof}
Let $F'$ be the right hand side of the equation.

\textbf{We first prove that $F_{k-1}(x,y,u,\ell,i,a,h)\le F'$.}

Let us consider a schedule 
$\mathcal{S}_1$ that realizes $G_{k-1}(x',y',\beta)$ and
$\mathcal{S}_2$ that realizes $F_{k-1}(y',y,u-\beta,\ell,i,a',h')$.
We suppose that the processor is idle during $[x,x']$.
We build a schedule with a empty set from $x$ to $x'$, 
with $\mathcal{S}_1$ from $x'$ to $y'$ and with 
$\mathcal{S}_2$ from $y'$ to $y$.

Obviously, the subsets $J(k-1,x,z)$ and $J(k-1,z,y)$ do not intersect,
so this is a feasible schedule, and
its cost is $F'$,
thus $F_{k-1}(x,y,u,f,\ell,i)\le  F' $.\\

\textbf{We now prove that $F'\le F_{k-1}(x,y,u,\ell,i,a,h)$.}

Let $\mathcal{O}$ to be an optimal schedule that realizes 
$F_{k-1}(x,y,u,\ell,i,a,h)$ such that $x'$ is the first 
starting time of the schedule and $y'$ is the completion
time of the first block of jobs in $\mathcal{O}$.
We split into two sub-schedules $\mathcal{S}_1\subseteq J(k-1,x',y')$ and 
$\mathcal{S}_2\subseteq J(k-1,y',y)$ such that
the value
of $x'$ is maximal and $y'$ is also maximal.



%
%
%


Then $y'-x'=a'+h'\times\frac{\ell}{i}$ for some value $a'=0,\ldots,a$
and $h'=0,\ldots,h$ by definition. Thus we can assume
that jobs in $\mathcal{S}_2$ have to be scheduled during
at most $(a-a')+(h-h')\times\frac{\ell}{i}$ unit time in $[y',y]$.
We claim that $x'$ is a release date by definition.


Moreover, we claim that all the jobs of $\mathcal{S}_2$ are released in $[y',y]$ 
and are completed before $y$.
If there exists a job in $\mathcal{S}_2$ which is not completed
at time $t$, then the removal of 
this job would lead to a lower energy consumption schedule for
$\mathcal{S}_2$ which contradicts the definition of
$F_{k-1}(y',y,|\mathcal{S}_2|,\ell,i,a-a',h-h')$.\\


Then the restriction $\mathcal{S}_1$ of $\mathcal{O}$ in 
$[x',y']$ is a schedule that meets all constraints
related to $G_{k-1}(x',y',|\mathcal{S}_1|)$. Hence its cost is
greater than $G_{k-1}(x',y',|\mathcal{S}_1|)$.
Similarly, the restriction  $\mathcal{S}_2$ of $\mathcal{O}$ to $[y',y]$ is a schedule that
meets all constraints related to 
$F_{k-1}(y',y,|\mathcal{S}_2|,\ell,i,a-a',h-h')$.

Thus $F'\le F_{k-1}(x,y,u,\ell,i,a,h)$.
\end{proof}

\section*{Proof of Proposition~\ref{prop_complexity_preemp}}

\begin{proof}
The values of $G_k(s,t,u)$ are stored in a multi-dimensional array of
size $O(|\Phi|^2 n^2)$.
Each value need $O(|\Phi| n^2 L^2 P^2 r(F))$ time to be computed where
$r(F)$ is the running time of the table $F_{k-1}(x,y,u,\ell,i,a,h)$.
Since we fix every value of $x,y,u,\ell,i,a,h$ in the minimization
step, the table $F$ doesn't need to be precomputed.
Then the running time is $O(n^2LP)$ for each value of $F$.
Therefore, the total running time of the dynamic programming
is $O(n^6L^9P^9)$.
Moreover, the values of $F_{k-1}(x,y,u,\ell,i,a,h)$ are stored in a multi-dimensional array (since we don't need to remember the $F_i$ values for $i<k-1$) of
size $O(n|\Phi|^2 L^2P^2 )=O(nL^6P^6)$.
\end{proof}

\section*{Proof of Proposition~\ref{prop_Gamma}}
\begin{proof}
Since $[s,x]$ is a reserved place for a job in the dynamic program, 
it is not necessary to consider
every value for $x$. Indeed, a job starts and ends belong the same set $\Theta_{a,b}$
for some value of $a$ and $b$ according to the proof of Proposition~\ref{theta}.
Then if $s$ is fixed, then $x$ can only take a 
reduced set of value which has a size of $O(n^3)$.
\end{proof}

\section*{Proof of Proposition~\ref{prop_complexity_non_preempt}}

\begin{proof}
The values of $E_k(s,x,t,u)$ are stored in a multi-dimensional array of
size $O(|\Gamma| |\Theta|^2 n^2)$.
Each value need $O(|\Gamma| |\Theta| n)$ time to be computed thanks to Proposition~\ref{prop_E} because there are $O(|\Theta|)$ values for $s'$,
$O(|\Gamma(s')|)$ values for $x'$ and $O(n)$ values for $\ell$.
Thus we have a total running time of $O(|\Gamma|^2|\Theta|^3n^3)$.
Moreover, the set of values $\Theta$ has a size of $O(n^4)$ because
$0\le s,t,i,h \le n$, thanks to Proposition~\ref{prop_Gamma},
$\Gamma(s)$ has a size of $O(n^3)$ for some value $s$.
This leads to an overall time complexity $O(n^{21})$.
\end{proof}

\section*{Proof of Proposition~\ref{prop_np_hard}}
\begin{proof}

In order to establish the $\mathcal{NP}$-hardness of $1|p_j=p,E|\sum_j w_j U_j$, we present a reduction from the {\sc Knapsack} problem which is known to be $\mathcal{NP}$-hard.
In an instance of the {\sc Knapsack} problem we are given a set $I$ of $n$ items. Each item $i\in I$ has a value $v_i$ and a capacity $c_i$.
Moreover, we are given a capacity $C$, which is the capacity of the knapsack, and a value $V$.

Given an instance of the {\sc Knapsack} problem, we construct an instance of $1|p_j=p,E|\sum_j w_j U_j$ as follows.
For each item $i$, $1\leq i \leq n$, we introduce a job $J_i$ with 
$r_i=\sum_{j=1}^{i-1}c_j$, $d_i=\sum_{j=1}^{i}c_j$, $w_i=v_i$ and $p_i=1$.
Moreover, we set the budget of energy equal to $E=C$.

We claim that the instance of the {\sc Knapsack} problem is feasible if and only if there is a feasible schedule for $1|p_j=p,E|\sum_j w_j U_j$ of total weighted throughput not less than $V$.

Assume that the instance of the {\sc Knapsack} is feasible.
Therefore, there exists  a subset of items $I'$ such that $\sum_{i\in I'}v_i\geq V$ and $\sum_{i\in I'}c_i\leq C$.
Then we can schedule the jobs in $I'$ with constant speed equals to 1. Their total energy consumption of this schedule is no more that $C$ since the instance of the Knapsack is feasible.
Moreover, their total weight is no less than $V$. 

For the opposite direction of our claim, assume there is a feasible schedule for\\
$1|p_j=p,E|\sum_j w_j U_j$ of total weighted throughput not less than $V$.
Let $J'$ be the jobs which are completed on time in this schedule.
Clearly, due to the convexity of the speed-to-power function, the schedule that executes the jobs in $J'$ with constant speed is also feasible.
Since the latter schedule is feasible, we have that $\sum_{j\in J'}(d_j-r_j)\leq C$.
Moreover, $\sum_{j\in J'}w_j\geq V$.
Therefore, the items which correspond to the jobs in $J'$ form a feasible solution for the {\sc Knapsack}. 
\end{proof}

\end{document}